%% file: main.tex
\pdfoutput=1
\newif\ifFull
\Fulltrue  
\newif\ifOrphan
\Orphanfalse
\documentclass{cccg25} 
\usepackage{graphicx}
\usepackage{url}
\usepackage{xcolor}
\usepackage{cite}
\usepackage{color}
\usepackage{calc}
\usepackage{xstring}
\usepackage[noend]{algorithmic}
\usepackage{algorithm}
\usepackage{subcaption}
\usepackage{todonotes}

\renewcommand{\emph}[1]{\textit{\textbf{#1}}}
\newcommand{\R}{\mathbf{R}}
\newcommand{\Eta}{\mathrm{H}}
\let\epsilon\varepsilon
\title{Entropy-Bounded Computational Geometry \\ Made Easier and 
       Sensitive to Sortedness\thanks{This research was supported in part by NSF grant 2212129.}}
\date{}

\ifFull
\author{David Eppstein 
\and Michael T. Goodrich 
\and Abraham M. Illickan 
\and Claire A. To \\[-6pt]
\and
Department of Computer Science, University of California, Irvine}
\else
\author{Anonymous author(s)}
\fi

\begin{document}

\maketitle

\begin{abstract}
We study entropy-bounded computational geometry, that is, 
geometric algorithms
whose running times depend on a given 
measure of the input entropy.
Specifically, we introduce a measure that we call \emph{range-partition entropy}, which
unifies and subsumes previous definitions of entropy used for sorting
problems and structural entropy used in computational geometry.
We provide simple algorithms for several problems, including 2D maxima,
2D and 3D convex hulls, and some visibility problems, and we show that they have running times depending on the range-partition entropy. 
\end{abstract}

\input{intro}

\section{Range-Partition Entropy}
Let $S=(p_1,p_2,\ldots,p_n)$ be the set of $n$ points in $\R^d$ input
in this given order, for constant $d\ge 1$.
Define a \emph{range partition} of $S$ to be a set, 
$\Pi=\{(S_1,R_1),(S_2,R_2),\ldots,(S_t,R_t)\}$, such that 
\begin{enumerate}
\item
The $S_i$'s form a partition of $S$, i.e., $S=\bigcup_{i=1}^t S_i$
and $S_i\cap S_j=\emptyset$ for $i\not=j$.
\item
The $R_i$'s are geometric \emph{ranges}, 
such as intervals in $\R$,
axis-aligned rectangles in $\R^2$, triangles in $\R^2$, or
tetrahedra in $\R^3$, such that
the set, $S_i$, is contained in the range, $R_i$.
\end{enumerate}
For example, if $S$ is a set of points in $\R^2$, 
then $\Pi$ could be a partition of $S$ into subsets contained 
in a set of axis-aligned rectangles, which form the ranges.

Given a sequence, $S$, of 
points in $\R^d$, and a range partition, 
$\Pi=\{(S_1,R_1),\ldots,(S_t,R_t)\}$, for $S$, we say
that $\Pi$ is \emph{respectful} if it satisfies the following
constraints:
\begin{enumerate}
\item
For each $i=1,2,\ldots, t$, $(S_i,R_i)$ satisfies
a given \emph{local property}, which depends only on $S_i$ and $R_i$.
\item
For each $i=1,2,\ldots, t$, $(S_i,R_i)$ satisfies
a given \emph{global compatibility} property, 
which can depend on the other pairs, $(S_j,R_j)$, for $j\not= i$.
\end{enumerate}

Given a set, $S$, of $n$ points in $\R^d$,
the \emph{entropy}, $\Eta(\Pi)$, of a partition,
$\Pi=\{(S_1,R_1),\ldots,(S_t,R_t)\}$,
of~$S$, is 
\[
\Eta({\Pi}) = - \sum_{i=1}^{t} 
               \left(\frac{|S_i|}{n}\right) \log \left(\frac{|S_i|}{n}\right).
\]
The \emph{range-partition entropy}, $\Eta(S)$,
of $S$ is the minimum $\Eta(\Pi)$ 
over all respectful partitions, $\Pi$.

We reformulate the sorting problem in this framework
by considering the input sequence, $S=(x_1,x_2,\ldots,x_n)$, 
to be points in $\R$, and we define a partition,
$\Pi=\{(S_1,R_1),\ldots,(S_t,R_t)\}$,
where each range, $R_i$, is an interval, $[a,b]\subset \R$.
In this case, the local property of each $(S_i,R_i)$ is that $S_i$ is
a consecutive subsequence of elements in $S$
given in sorted order, and the global property
is that the $R_i$ ranges are disjoint, i.e., $R_i\cap R_j=\emptyset$
for $i\not=j$.
Accordingly,
the minimum entropy, $\Eta(\Pi)$, for all respectful partitions,
$\Pi$, is determined by a partition of the input sequence into maximal
non-decreasing or non-increasing runs.
Thus, our framework subsumes the notion of entropy used
for the sorting problem.
Also, as we discuss in more detail in subsequent sections,
it also subsumes the structural entropy introduced
by Afshani {\it et al.}~\cite{afshani}.


\section{2D Maxima Set}
The problem we study in this section
is to find the maxima set of a given set of points in the
plane, where a point is considered \emph{maximal} if no other
point \emph{dominates} it, having both greater $x$- and $y$-coordinates.
Our algorithm for this problem 
is a simple a variant of Kirkpatrick and Seidel's
``marriage-before-conquest'' method~\cite{10.1145/323233.323246},
which was also studied by Afshani {\it et al.}~\cite{afshani},
except our algorithm takes advantage of near-sortedness.  

As preprocessing, we find the point $p_{\max}$
in $S$ with maximum $x$-coordinate, guaranteed to be a maximum point. We prune from $S$ any point dominated
by $p_{\max}$.
Our remaining algorithm
begins by checking in linear time 
if the input is sorted (e.g., by $x$- or $y$-coordinates).
If so, it computes the maxima set in linear time by a simple plane-sweeping
stack algorithm.
Otherwise, in linear time, we
partition the points into left and right subsets
based on the median $x$-coordinate using a stable method, we find and add to the maxima set
a point with largest
$y$-coordinate in the right subset, and we remove all points dominated by this point.
Then our
algorithm recursively solves the maxima set problem
for the remaining points in the left and right subsets.
See Algorithm~\ref{alg:maxima} in Appendix A.


Given an input sequence of $n$ points, $S$, in $\R^2$,
define the constraints for 
a partition, $\Pi=\{(S_1,R_1),\ldots,(S_t,R_t)\}$, where
each $R_i$ is an axis-aligned rectangle, to be
\emph{respectful} in the 
context of computing the maxima set for $S$
(Figure~\ref{fig:rect}):
\begin{enumerate}
\item
The \emph{local property} for each $(S_i,R_i)$ is that $R_i$
is an axis-aligned rectangle containing $S_i$
and either $S_i$
forms a sorted subsequence in $S$ or the upper right corner of $R_i$ is dominated by some point in $S$.
\item
The \emph{global compatibility} property is that,
for $i,j=1,2,\ldots,t$, if $R_i$ is 
in not dominated by a single point in $S$ (which means $S_i$ is in sorted order in $S$), then it
will not intersect another range, $R_j$, for $j\not= i$.
\end{enumerate}

This generalizes the notion of structurally respectful partitions from Afshani {\it et al.}~\cite{afshani}, which is equivalent to 
the local condition above for unsorted subsets.
It is easy to construct inputs whose structural entropy is much higher than their range-partition 
entropy, such as a set of $n$ maxima points given in sorted order.
%
%

\begin{figure}[hbt]
    \centering
\includegraphics[width=0.5\linewidth,page=1]{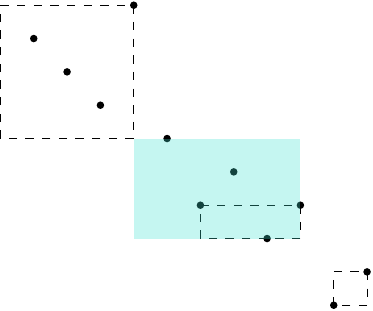} 
\\ (a) \\
\includegraphics[width=0.5\linewidth,page=2]{figures/maxima-partitions.pdf} 
\\ (b)
\caption{Respectful partitions.
The points in the blue shaded rectangle are sorted among themselves.
(a) A respectful partition 
 without making use of the sorted set. 
The entropy is $\frac{4}{11}\log\frac{11}{4}+\frac{1}{11}\log11+\frac{1}{11}\log11+\frac{3}{11}\log\frac{11}{3}+\frac{2}{11}\log\frac{11}{2}\approx2.118$.
(b) A respectful partition making use of sets of both types. 
The entropy is $\frac{4}{11}\log\frac{11}{4}+\frac{5}{11}\log\frac{11}{5}+\frac{2}{11}\log\frac{11}{2}\approx1.495$.
}
        \label{fig:rect} 
\end{figure}

 

\begin{theorem}
Given a sequence, $S$, of $n$ points in $\R^2$,
the \textsc{2DMaximaSet} algorithm runs in $O(n(1+\Eta(S)))$ time,
where $\Eta(S)$ is the range-partition entropy of~$S$.
\end{theorem}
\begin{proof}
Let us analyze 
2DMaximaSet via an accounting argument
where a constant amount of work in our algorithm costs one cyber-dollar.
Let $\Pi=\{(S_1,R_1),\ldots,(S_t,R_t)\}$ be 
a respectful partition of $S$ with minimum range-partition
entropy, $\Eta(S)$.
Since $n(1+\Eta(S)) =n+n(\Eta(S))$, 
let us focus on the term
\[
n(\Eta(S)) = \sum_{S_k\in \Pi}|S_k|\log(n/|S_k|). 
\]
Thus, after charging each point in $S$ one cyber-dollar, we 
can show that 2DMaximaSet runs in time $O(n(1+\Eta(S)))$ by showing that the processing 
we perform for each 
set, $S_k\in\Pi$, contributes at most $O(|S_k|\log (n/|S_k|)+1)$ 
additional cyber-dollars to the running time of our algorithm. Let $T$ denote the
recursion tree for 2DMaximaSet,
where each node, $v$, of $T$ corresponds to a recursive call. Each node $v$ of $T$ is associated with an interval, $I_v=[a_v,b_v]$, of 
$x$-coordinates for points of $S$ between discovered maximal points for ancestors
of $v$ in $T$ (or with $b_v=p_{\max}$ for each node, $v$,
on the right spine of $T$).

So, consider a subset, $S_k\in\Pi$ which has the local property for unsorted sets. Say that $S_k$ \emph{covers} $v$ if the $x$-range for $R_k$ 
spans $I_v$ but not $I_{\mathrm{parent}(v)}$.

$S_k$ is contained in an axis-aligned box, $R_k$, 
that is strictly below the staircase. 
Thus, for any node, $v$ in $T$, if $S_k$ covers $v$, then all the
points of $S_k$ are removed from any recursive calls associated
with $v$ or its descendants in $T$, because they are all dominated by
the upper-right corner of $R_k$, which in turn is dominated by, or just is the highest point to the right side of the interval and was previously discovered to be a maxima point. If $S_k$ covers a node, $v$ in $T$, then it 
does not participate in any recursive calls for descendants of
$v$ in $T$.
Thus, the maximum number of points in $S_k$ that survive to level
$j$ in $T$ is $O(\min\{|S_k|,\lceil n/2^j\rceil\}$). 

Now consider a set
$S_k$ that has the local property for a sorted subset and its range, $R_k$. $S_k$ also
satisfies the global compatibility condition for $S_k$.
That is, $R_k$ is an axis-aligned box
that contains $S_k$ such that the points of $S_k$ are given
in sorted order in $S$ and there is no other range intersecting $R_k$. We will consider $S_k$ in two parts. Let $p$ be the rightmost point that is above $R_k$.  $S_k^{(\ell)}$ is the subset of $S_k$ that is dominated by $p$. $S_k^{(r)}$ is the remaining subset. Let $R_k^{(\ell)}$ (similarly $R_k^{(r)}$) be the rectangle that has the same $y$-range as $R_k$ but the minimum $x$-range such that it still contains $S_k^{(\ell)}$ (similarly $S_k^{(r)}$). Say that $S_k^{(\ell)}$ (similarly $S_k^{(r)}$) \emph{covers} $v$ if the $x$-range for $R_k^{(\ell)}$ (similarly $R_k^{(r)}$) 
spans $I_v$ but not $I_{\mathrm{parent}(v)}$.

If $S_k^{(\ell)}$ covers a node $v$ in $T$, $b_v$ is $p$ or above $p$ and all the points of $S_k^{(\ell)}$ have been pruned away and are removed from any recursive calls associated with $v$ and its descendants. Thus, the maximum number of points in $S_k^{(\ell)}$ that survive to level
$j$ in $T$ is $O(\min\{|S_k^{(\ell)}|,\lceil n/2^j\rceil\})$. 

If $S_k^{(r)}$ covers a node $v$ in $T$, $a_v$ and $b_v$ are in $S_k^{(r)}$, and no points above or below it remain in the recursive call associated with $v$. Since $S_k^{(r)}$ as a subset of $S_k$ is sorted by $x$-coordinate, during the recursive call associated with $v$, our algorithm recognizes this and computes the maxima set in linear time and there are no more recursive calls. Thus, the maximum number of points in $S_k^{(r)}$ that survive to level
$j$ in $T$ is $O(\min\{|S_k^{(r)}|,\lceil n/2^j\rceil\})$.

Therefore, the maximum number of points in $S_k$ that survive to level
$j$ in $T$ is $O(\min\{|S_k^{(\ell)}|,\lceil n/2^j\rceil\})+O(\min\{|S_k^{(r)}|,\lceil n/2^j\rceil\}) =O(\min\{|S_k|,\lceil n/2^j\rceil\})$. 

Let $n_j$ denote the total number of points in $S$ that survive to level
$j$ in $T$, and note that the total time (in cyber-dollars)
charged to the 2DMaximaSet algorithm is proportional to 
$\sum_{j=0}^{\lceil\log n\rceil} n_j$.
The proof follows, then,
by the following:
\begin{equation*}
\begin{split}
    \sum_{j=0}^{\lceil\log n\rceil} n_j \leq \sum_{j=0}^{\lceil \log n\rceil} \sum_k
    \min\left\{|S_k|,O\left(\frac{n}{2^j}\right)\right\}\\ \leq
    \sum_k \sum_{j=0}^{\lceil \log n\rceil}
    \min\left\{|S_k|,O\left(\frac{n}{2^j}\right)\right\}\\ \leq
    \sum_k O(|S_k|\lceil\log(n/|S_k|)\rceil + |S_k| +\frac{|S_k|}{2}
    + \frac{|S_k|}{4} + \dots + 1)\\ \leq \sum_k O(|S_k|(\lceil
    \log(n/|S_k|)\rceil +2)) ~~~\in~ O(n(\Eta(\Pi)+1))
\end{split} 
\end{equation*} 
\end{proof}

\section{2D Convex Hull}
Here we apply a similar modification to leverage sortedness in 2D convex hulls. Finding the convex hull in $\R^2$ involves identifying the smallest convex polygon that encloses all input points and returning the ordered subset of points that lie on the boundary. As with maxima sets, our algorithm follows Kirkpatrick and Seidel's approach as presented in \cite{afshani} with a modification to check for sortedness. 

The algorithm proceeds as follows. We begin by computing the upper hull. To do this, we first prune all points below the line connecting the leftmost and rightmost points of the input. Check if the points sorted. If a they are sorted, find the convex hull in linear time using Graham's scan~\cite{GRAHAM1972132}. Otherwise, the points are partitioned into two subsets based on the median $x$-coordinate, using a stable method. Next, we identify the two points that form the edge of the upper hull and intersect the vertical line at the median $x$-coordinate. This step can be done in linear time using the same method as Kirkpatrick and Seidel~\cite{kirkpatrickseidelconvex}. All points below this edge are not in the convex set and can be pruned.  Recursively solve each half. Finally, the convex hulls of both subsets are concatenated to obtain the upper hull. The lower hull is computed similarly.
See Algorithm~\ref{alg:2dhull} in Appendix A.

Given an input sequence of $n$ points, $S$, in $\R^2$,
let us define the constraints for 
a partition, $\Pi=\{(S_1,R_1),\ldots,(S_t,R_t)\}$, where
each $R_i$ is a triangle, to be
\emph{respectful} in the 
context of computing the convex hull for $S$
(see Figure~\ref{fig:conv1}):

\begin{enumerate}
\item
The \emph{local property} for each $(S_i,R_i)$ is that $R_i$
is a triangle containing the points of $S_i$
and either $S_i$
forms a sorted subsequence in $S$ or $R_i$ lies under the convex hull of $S$.
\item
The \emph{global compatibility} property is that,
for $i,j=1,2,\ldots,t$, if $R_i$ does not lie under the convex hull of $S$ (which means $S_i$ is in sorted order in $S$), then it
will not intersect another range, $R_j$, for $j\not= i$.
\end{enumerate}

As for the sorting problem, this generalizes the notion of structurally respectful partitions of Afshani {\it et al.}~\cite{afshani}, which is equivalent to 
the local condition above for unsorted subsets. Inputs whose structural entropy is much higher than their range-partition 
entropy include sets of $n$ points on a circle given in sorted order.



\begin{figure}[hbt]
    \centering
\begin{tabular}{cc}
    \includegraphics[width=0.45\linewidth,page=1]{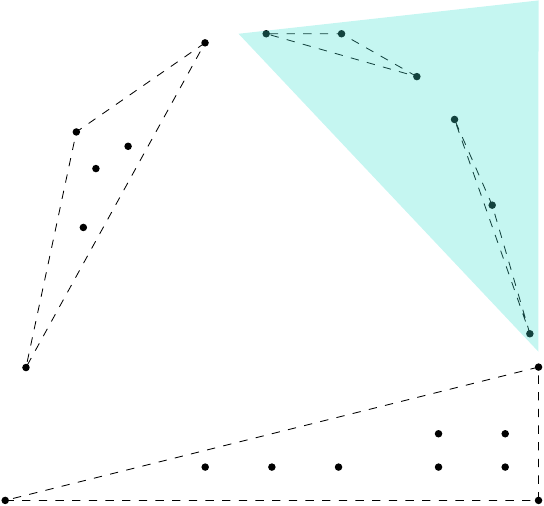} 
&
    \includegraphics[width=0.45\linewidth,page=2]{figures/hull-partitions.pdf} \\
(a) &
(b)
\end{tabular}
\caption{Respectful partitions for convex hulls.
(a) The points in the blue triangle are sorted among themselves but this partition doesn't take advantage of that. The entropy is $\frac{6}{22}\log\frac{22}{6}+\frac{3}{22}\log\frac{22}{3}+\frac{3}{22}\log\frac{22}{3}+\frac{10}{22}\log\frac{22}{10}\approx1.811$.
    (b) This partition takes advantage of the sortedness. The entropy is $\frac{6}{22}\log\frac{22}{6}+\frac{6}{22}\log\frac{22}{6}+\frac{10}{22}\log\frac{22}{10}\approx1.540$.}
    \label{fig:conv1}
\end{figure}

\begin{theorem}
    Given a set, $Q$, of $n$ points in $\R^2$ the 2DConvexHull algorithm runs in $O(n(1+ \Eta(S)))$ time, where $\Eta(S)$ is the range-partition entropy of $S$.
\end{theorem}
\begin{proof}
Let us analyze 
2DConvexHull via an accounting argument
where a constant amount of work in our algorithm costs one cyber-dollar.
Let $\Pi=\{(S_1,R_1),\ldots,(S_t,R_t)\}$ be 
a respectful partition of $S$ with minimum range-partition
entropy, $\Eta(S)$.
Since $n(1+\Eta(S)) =n+n(\Eta(S))$, 
let us focus on the term
\[
n(\Eta(S)) = \sum_{S_k\in \Pi}|S_k|\log(n/|S_k|). 
\]
Thus, after charging each point in $S$ one cyber-dollar, we 
can show that 2DConvexHull runs in time $O(n(1+\Eta(S)))$ by showing that the processing 
we perform for each 
set, $S_k\in\Pi$, contributes at most $O(|S_k|\log (n/|S_k|)+1)$ 
additional cyber-dollars to the running time of our algorithm. Let $T$ denote the
recursion tree for 2DConvexHull,
where each node, $v$, of $T$ corresponds to a recursive call. Each node $v$ of $T$ is associated with an interval, $I_v=[a_v,b_v]$, of 
$x$-coordinates for points of $S$ between discovered convex hull points for ancestors
of $v$ in $T$ (or with $b_v=p_{\max}$, the point with the largest $x$-coordinate for each node, $v$,
on the right spine of $T$).

So, consider a subset, $S_k\in\Pi$ which has the local property for unsorted sets.

$S_k$ is contained in a triangle, $R_k$, 
that is strictly below the convex hull. We will consider $S_k$ in two parts. Let $p$ be the vertex of $R_k$ that is opposite the lowest edge.  $S_k^{(\ell)}$ is the subset of $S_k$ that is to the left of $p$. $S_k^{(r)}$ is the remaining subset. Let $R_k^{(\ell)}$ (similarly $R_k^{(r)}$) be the triangle on the left (similarly right) obtained by splitting $R_k$ along $x=p_x$. $S_k^{(\ell)}$ is contained in $R_k^{(\ell)}$ and$S_k^{(r)}$ is contained in $R_k^{(r)}$. Say that $S_k^{(\ell)}$ (similarly $S_k^{(r)}$) \emph{covers} a node $v$ of $T$ if the $x$-range for $R_k^{(\ell)}$ (similarly $R_k^{(r)}$) 
spans $I_v$ but not $I_{\mathrm{parent}(v)}$.

If a node $v$ is covered by $S_k^{(\ell)}$ (similarly $S_k^{(r)}$), the edge between the convex hull points with the $x$-coordinate $a_v$ and $b_v$ lie above $R_k^{(\ell)}$ (similarly $R_k^{(r)}$). During the recursive call corresponding to $v$, all points below this edge are pruned and all points from  $S_k^{(\ell)}$ (similarly $S_k^{(r)}$) are removed from descendants of $v$.
Thus, the maximum number of points in $S_k^{(\ell)}$ (similarly $S_k^{(r)}$) that survive to level
$j$ in $T$ is $O(\min\{|S_k^{(\ell)}|,\lceil n/2^j\rceil\})$ (similarly $O(\min\{|S_k^{(r)}|,\lceil n/2^j\rceil\})$).
Therefore, the maximum number of points in $S_k$ that survive to level
$j$ in $T$ is $O(\min\{|S_k^{(\ell)}|,\lceil n/2^j\rceil\})+O(\min\{|S_k^{(r)}|,\lceil n/2^j\rceil\}) =O(\min\{|S_k|,\lceil n/2^j\rceil\})$. 

Now consider a set
$S_k$ that has the local property for a sorted subset and its range, $R_k$. $S_k$ also
satisfies the global compatibility condition for $S_k$.
That is, $R_k$ is triangle
that contains $S_k$ such that the points of $S_k$ are given
in sorted order in $S$ and there is no other range intersecting $R_k$. We will consider $S_k$ in five parts. Since $R_k$ is a triangle, there can be at most two continuous sections of the convex hull that are strictly above $R_k$. Let $S_k^{(\ell)}$, $S_k^{(m)}$, and $S_k^{(r)}$ be the subsets of $S_k$ that is to the left, middle, and right of these sections respectively. Let $S_k^{(m\ell)}$ and $S_k^{(mr)}$ be the subsets of $S_k$ that are below each of these sections. Let $R_k^{(\ell)}$, $R_k^{(m\ell)}$, $R_k^{(m)}$, $R_k^{(mr)}$, and $R_k^{(r)}$ be the intersection of $R_k$ with the corresponding $x$-ranges. Say that $S_k^{(\mathcal{X})}$\emph{covers} $v$ if the $x$-range for $R_k^{(\mathcal{X})}$ spans $I_v$ but not $I_{\mathrm{parent}(v)}$.

If $S_k^{(\mathcal{X})}$ is $S_k^{(m\ell)}$ or $S_k^{(mr)}$, and covers a node $v$ in $T$, $a_v$ and $b_v$ are points in the convex hull that are above $S_k^{(\mathcal{X})}$ and the points of $S_k^{(\mathcal{X})}$ will be pruned away and are removed from any recursive calls associated with $v$ and its descendants. Thus, the maximum number of points in $S_k^{(\mathcal{X})}$ that survive to level
$j$ in $T$ is $O(\min\{|S_k^{(\mathcal{X})}|,\lceil n/2^j\rceil\})$. 

If $S_k^{(\mathcal{X})}$ is $S_k^{(\ell)}$, $S_k^{(m)}$, or $S_k^{(r)}$, and  covers a node $v$ in $T$, $a_v$ and $b_v$ are in $S_k^{(\mathcal{X})}$, and no points above or below it remain in the recursive call associated with $v$. Since $S_k^{(\mathcal{X})}$ as a subset of $S_k$ is sorted by $x$-coordinate, during the recursive call associated with $v$, our algorithm recognizes this and computes the convex hull in linear time and there are no more recursive calls. Thus, the maximum number of points in $S_k^{(\mathcal{X})}$ that survive to level
$j$ in $T$ is $O(\min\{|S_k^{(\mathcal{X})}|,\lceil n/2^j\rceil\})$.

Therefore, the maximum number of points in $S_k$ that survive to level
$j$ in $T$ is $O(\min\{|S_k^{(\ell)}|,\lceil n/2^j\rceil\})+O(\min\{|S_k^{(m\ell)}|,\lceil n/2^j\rceil\})+O(\min\{|S_k^{(m)}|,\lceil n/2^j\rceil\})+O(\min\{|S_k^{(mr)}|,\lceil n/2^j\rceil\})+O(\min\{|S_k^{(r)}|,\lceil n/2^j\rceil\}) =O(\min\{|S_k|,\lceil n/2^j\rceil\})$. 

Let $n_j$ denote the total number of points in $S$ that survive to level
$j$ in $T$, and note that the total time (in cyber-dollars)
charged to the 2DConvexHull algorithm is proportional to 
$\sum_{j=0}^{\lceil\log n\rceil} n_j$. The proof follows, then, since
$    \sum_{j=0}^{\lceil\log n\rceil} n_j  \in O(n(\Eta(\Pi)+1)))$.
\end{proof}

\section{Visibility and Lower Envelope Problems}
In this section, we show how to
apply range-partition entropy to analyzing some
\textit{lower envelope} and \textit{visibility polygon} problems.
In these applications, there is no structural components, however, so
the range-partition entropy is the same as
the entropy in these applications
as the entropy used for the sorting problem;
hence, our algorithms 
also provide geometric applications of the entropy used for 
the sorting problem.

\paragraph{Lower envelope.} 
Given a set of $\rho$ disjoint monotone polygonal chains,
$\mathcal{S} = \{S_1, S_2, \ldots, S_\rho\}$, of total size $n$, the first problem
we study is to compute
the piecewise-linear lower envelope of the chains in $\mathcal{S}$.

Let us define the constraints for 
a partition, $\Pi=\{(S_1,R_1),\ldots,(S_t,R_t)\}$, where
each $R_i$ is the polygonal chain formed by $S_i$, to be
\emph{respectful} in the 
context of computing the lower envelope for $S$
(see Figure~\ref{fig:lower}):
\begin{enumerate}
\item
The \emph{local property} for each $(S_i,R_i)$ is that $R_i$
is the polygonal chain of $S_i$ and 
$S_i$ forms a sorted subsequence in $S$.
\item
The \emph{global compatibility} property is that,
for $i,j=1,2,\ldots,t$, $R_i$ will not intersect another range, $R_j$, for $j\not= i$
\end{enumerate}
In this context, each $R_i$ is simply the same as $S_i$, and is included only for consistency with the definition of a respectful partition used across problems. Since this application does not depend on its structural components, $R_i$ plays no additional role in the analysis.

Let $\mathcal{A}$ be
a stack-based mergesort algorithm, like
\textsc{TimSort}~\cite{peters2015timsort}, 
be a sorting algorithm that leverages monotonic runs of an input sequence, 
$X$, of $n$ elements to run in $O(n(1+\Eta(X)))$ time, where
$\Eta(X)$ is the range-partition entropy of $X$ (which is the same
as the entropy previously studied for the sorting problem).
We show how to adapt $\mathcal{A}$ to computing the lower envelope of 
the disjoint monotone chains in $\mathcal{S}$
in $O(n(1+\Eta(S)))$ time.
The algorithm, $\mathcal{A}$, works by maintaining a stack of maximal non-decreasing
runs and merging consecutive pairs of them using the merge 
algorithm from mergesort for merging two sorted sequences according to rules
based on their sizes that leads 
to the $O(n(1+\Eta(X)))$ running time.

Our first adaptation of the algorithm, $\mathcal{A}$, 
is to change the merge algorithm to be a merge of two lower envelopes,
where we merge the sequences by $x$-coordinates 
and determine, at each $x$-coordinate, the segment with the smallest $y$-value,
compressing the sequence to eliminate segment endpoints not in the lower envelope.
We assign each lower envelope a \emph{weight} calculated as the sum 
of all contributing segments to the merged result, 
rather than the current number of segments in the lower envelope. 

The adapted version of the algorithm, $\mathcal{A}$, maintains a stack of active chains.
Suppose that the stack contains chains $S_1, S_2, \dots, S_k$, with corresponding weights $w_1, w_2, \dots, w_k$. 
We iterate through the input chains, 
pushing a new chain onto the stack at each iteration, as dictated by $\mathcal{A}$. 
After each insertion, we invoke the merge procedure to restore $\mathcal{A}$'s
invariant conditions on the stack weights. 
For example, if $\mathcal{A}$ were TimSort, then
we would ensure that the following invariants holds at the end of each iteration:
\[
w_{i+2} \geq w_{i+1} + w_i, \tag{1}
\]
\[
w_{i+1} \geq w_i. \tag{2}
\]
These conditions guarantee that the weights 
of sequences on the stack in TimSort grow at least as fast as the Fibonacci numbers, 
which ensures that the height of the stack remains logarithmic~\cite{auger2019}.

To merge two chains as dictated by $\mathcal{A}$, 
we perform a linear scan over their segments, 
selecting the smallest $y$-value at each $x$-coordinate. 
Since each input chain is $x$-monotone and disjoint, 
the merging process runs in linear time respect to 
the total number of contributing input segments, up to a constant factor. 
This is analogous to merging two sorted sequences.
See Figure~\ref{fig:lower}.

\begin{figure}[hbt]
        \centering
        \includegraphics[width=.6\linewidth,page=1]{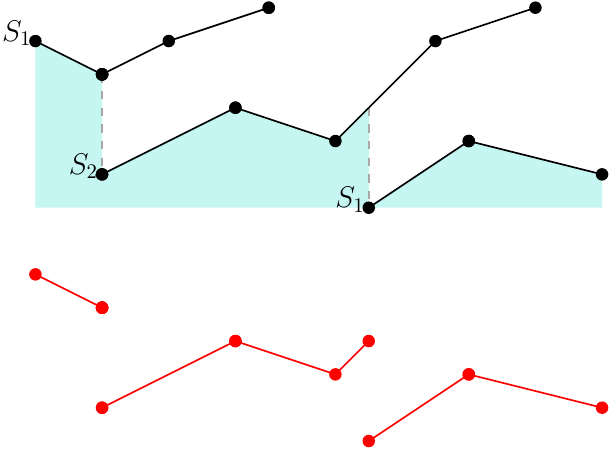}
        \caption{Merging two sets of disjoint monotone chains, where $S_1$ contains of two sequences and $S_2$ contains one. 
}
\label{fig:lower}
\end{figure}

Although the current envelope may contain fewer segments due to pruning of segments 
after some merges, the total weight still reflects the work that would 
have been done by $\mathcal{A}$ over the 
original input segments, preserving the same invariants used in $\mathcal{A}$'s
merge strategy.



\begin{theorem}
\label{thm:lower}
Given an set, $S$, of $n$ line segments partitioned into 
$\rho$ disjoint monotone chains, we compute the lower envelope 
of the chains in $S$ in $O(n(1+\Eta(S)))$ time, 
where $\Eta(S)$ is the range-partition entropy of~$S$.
\end{theorem}
\begin{proof}
The proof follows from the analysis of the algorithm, $\mathcal{A}$.
For example, Theorem~1 in the analysis of TimSort 
by Auger {\it et al.}~\cite{auger2019} 
shows that TimSort runs in $O(n + n\Eta(X))$ time on an input, $X$, 
of length $n$, where $\Eta(X)$ is the range-partition entropy of $X$.
the run-length distribution. 
Our adaptation of $\mathcal{A}$ into a lower envelope algorithm 
can be seen as analogous to 
$\mathcal{A}$ in the sense that each merge in our lower envelope algorithm runs
in time proportional to the time needed to perform the merges in $\mathcal{A}$.
Thus, the running time for our lower-envelope algorithm is 
$O(n(1+\Eta(S)))$.
\end{proof}

\paragraph{Visibility polygon.}  
In the next problem we consider,
we are
given a convex polygon, $P$, containing a set of convex polygonal obstacles,
with total complexity $n$,
and a query point, $q$, in~$P$'s interior, the problem is to 
compute the region of~$P$ that is visible from~$q$. 
We assume that $P$ and the convex obstacles in $P$ are provided as 
a set of $\rho$ disjoint convex chains,
$S = \{S_1, S_2, \ldots, S_\rho\}$, where each chain represents either 
the outer boundary or an obstacle.

This
visibility problem can be reduced to computing a 
lower envelope of monotone disjoint sequences, for which we provided an algorithm
above. 
Namely,
at each point along the chains, the visible segment of the polygon corresponds 
to the segment with the minimum radius with respect to~$q$ 
rather than the minimum $y$-value in the lower envelope setting. 
For each chain, we prune the vertices outside this range by 
locating the two endpoints of the visible window from~$q$. 
Once pruned, the remaining chains can be merged using a 
linear scan based on the angle each point makes with respect to~$q$. 
As before, we maintain the \emph{weight} of each chain as the sum of all contributing segments, which increases after each merge. 
See Figure~\ref{fig:vis}.

\begin{figure}[hbt]
    \centering
        \includegraphics[width=.5\linewidth,page=1]{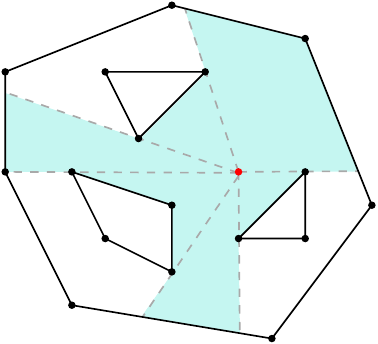}
\caption{The visibility polygon of a point among disjoint convex chains. 
}
        \label{fig:vis}
\end{figure}



\begin{theorem}
Suppose we are
given a polygon, $P$, and convex obstacles in its interior,
which are given as a collection, $S$, of $\rho$ disjoint convex 
chains consisting of a total of $n$ vertices, where each 
chain represents the outer boundary or an obstacle.
Then our visibility polygon algorithm 
runs in $O(n(1+\Eta(S)))$ time, 
where $\Eta(S)$ is the range-partition entropy of~$S$.
\end{theorem}
\begin{proof}
The proof for the running time follows 
Theorem~\ref{thm:lower}.
Instead of selecting the minimum $y$-value at each $x$-value in the 
proof of Theorem~\ref{thm:lower}, however,
we select the minimum radius at each angle with respect to the point $q$.
Thus, our visibility polygon algorithm runs in $O(n(1+\Eta(S)))$ time.
\end{proof}

\bibliographystyle{plainurl}

\bibliography{ref}

\clearpage
\appendix

\section{Pseudocode}
In this section, we provide pseudocode descriptions of our 2D algorithms.

\begin{algorithm}
    \caption{2DMaximaSet($S$):
Given a set of $n$ points, $S$, in $\R^2$, 
compute the maxima set, $X$, of $S$.
\label{alg:maxima}
    }
    \begin{algorithmic}[1]
     \IF {$S$ is sorted \textbf{or} $n\le 1$} 
     \STATE Compute the maxima set for $S$ in $O(n)$ time, add
            the maxima points to $X$, and {\bf return}.
     \ENDIF
        \STATE Partition into $S_\ell$ and $S_r$ by median $x$-coordinate using a stable method.
        \STATE Compute $q \in S_r$ with maximum $y$-coordinate.
        \STATE Add $q$ to the output set $S$, and
        delete $q$ and prune all points in $S_\ell$ 
        and $S_r$ that are dominated by $q$.   
        \STATE Recursively compute the maxima set for $S_\ell$, and
        add the maxima set of $S_\ell$ to $X$.
        \STATE Recursively compute the maxima set for $S_r$, and
        add the maxima set of $S_r$ to $X$.
    \end{algorithmic}
\end{algorithm}

\begin{algorithm}
    \caption{2DConvexHull($S$):
Given a set of $n$ points, $S$, in $\R^2$, 
compute the convex hull, $X$, of $S$.
\label{alg:2dhull}
    }
    \begin{algorithmic}[1]
     \STATE Prune all the points strictly below the line between the leftmost and rightmost points of $S$.
     \IF {$S$ is sorted \textbf{or} $n\le 1$} 
     \STATE Compute the convex hull of $S$ in $O(n)$ time, and {\bf return}.
     \ENDIF
        \STATE Partition into $S_\ell$ and $S_r$ by median $x$-coordinate $m$ using a stable method.
        \STATE Identify edge $qq'$ of the convex hull that intersects the line $x=m$ and prune the points below it.
        \STATE Recursively compute the convex hull for $S_\ell$.
        \STATE Recursively compute the convex hull for $S_r$.
        \STATE Concatenate and return.
    \end{algorithmic}
\end{algorithm}

\section{3D Convex Hull}
\label{sec:3dhull}
In this section we simplify an algorithm of Afshani et al.~\cite{afshani} for convex hulls of points in $\mathbb{R}^3$, by finding a simple randomized replacement for a key subroutine used by Afshani et al. For the set of points $S\subset \mathbb{R}^3$, we say that a partition $\Pi$ of $S$ is structurally respectful if for each $S_k\in \Pi$, there is a tetrahedron $\Delta_k$ that contains all the points of $S_k$ and lies under the convex hull. The structural entropy $\Eta(\Pi)$ of a partition $\Pi$ is defined as $\sum_{S_k\in \Pi}(|S_k|/n)\log(n/|S_k|)$. The entropy of $S$ is the minimum over all structurally respectful partitions of $S$.

\begin{theorem}[Theorem 3.9 in~\cite{afshani}]
    Given a set, $S$, of $n$ points in $\R^3$, algorithm \texttt{hull3d} runs in $O(n(1+ \Eta(S)))$ time, where $\Eta(S)$ is the range-partition entropy of $S$.
\end{theorem}
In this algorithm, Afshani et al. assume an efficient subroutine to partition the points using Matoušek’s partition theorem~\cite{Avis1992} or the recursive use of the eight-sectioning theorem~\cite{octants}. The first of these methods runs in the required $O(n\log n)$ time, but the second requires $O(n^6\log n)$ time. Our contribution for this section is to show how to find an approximate eight-section in expected linear time by using random sampling and combinatorial partitioning to divide the space into fair enough octants. Recursively doing this gives us the required $O(n\log n)$ time for the subroutine. For completeness, we provide the description of the convex hull algorithm using this subroutine in \cref{3dalg}.
\begin{algorithm}
    \caption{Given a set $Q$ of $n$ points in $\mathbb{R}^3$, the algorithm finds a balanced partition of the points by dividing the space into octants as a subroutine to compute the 3D convex hull of $Q$.}
    \textbf{Algorithm} EightPartition($Q$):
    \begin{algorithmic}[1]
        \STATE Define $S$ be a subset of points, chosen uniformly at random from $Q$.
        \STATE Let $|S| \approx n^{1/10}$.

        \STATE Define $P \gets \emptyset$ to store the plane partitions
        \STATE Define $O \gets \emptyset$ to store the octant partitions
        
        \FORALL{triplet of points $(p_1, p_2, p_3) \in S$}
            \STATE Define the plane $\pi_i$ passing through $p_1, p_2, p_3$.
            \FORALL{perturbations of the plane $\pi_i$}
                \STATE Classify all points as \textit{above}, \textit{on}, or \textit{below} $\pi_i$.
                \STATE Add this partition to $P$.
            \ENDFOR
        \ENDFOR
    
        \FORALL{triplet of planes $(\pi_1, \pi_2, \pi_3) \in P$}
            \STATE Add the resulting partition to $O$.
        \ENDFOR
    
        \FORALL{octant partitions}
            \STATE Compute the number of points in each region.
            \IF{the partition is fair enough}
                \STATE Return this partition.
            \ENDIF
        \ENDFOR
    \end{algorithmic}
\end{algorithm}
\begin{lemma}\label{lem:octant}
    For $n$ points in $\mathbb{R}^3$, we can find a partition of this set into octants such that each octant has between $\frac{n}{16}$ and $\frac{3n}{16}$ points, in expected linear time. 
\end{lemma}
\begin{proof}
    First, we sample the points with probability $p = n^{-9/10}$. By linearity of expectation, the expected sample size is $\mu = n^{1/10}$. By the multiplicative Chernoff bound with $\delta = 0.5$, the probability that the sample size will be between $\frac{n^{1/10}}{2}$ and $\frac{3n^{1/10}}{2}$ is at least $1 - 2e^{-\frac{n^{1/10}}{12}}$. We repeat the sampling process until the sample size falls within this range.

Partitioning planes are generated by selecting three points from this subset. These three points define a plane in 3D space. We slightly perturb this plane such that each of the three points lies either above, on, or below the plane. This results in 27 distinct ways in which the three points can define the plane. Since there are $O(n^{1/10})$ points in the sample, the number of ways to choose three points from this sample is $O(n^{3/10})$, and for each combination of three points, there are 27 distinct partitions. Therefore, we obtain at most $27n^{3/10}$ unique partitions from single planes. To construct a full division into octants, three planes are combined, resulting in up to $O(n^{9/10})$ possible partitions of the space into eight regions. 

The algorithm examines these partitions to find one that distributes the points evenly across all octants. We are guaranteed that there is such a fair partition due to \cite{octants}. Because checking whether a single partition is balanced takes $O(n^{1/10})$ time, the runtime is $O(n)$.

Consider a subset of the original set of points that has $cn$ points. The expected number of points that are sampled from this subset is $cn^{1/10}$. By a special case of Hoeffding's inequality~\cite{Hoeffding01031963}, the probability that the fraction of points that are in this subset among all points in the sample deviates from~$c$ by more than $1/16$ is less than $2e^{\frac{-n^{1/10}}{128}}$.

After getting a partition that evenly divides the sample into octants, we check if the same partition also divides all the points evenly with every octant having between $\frac{n}{16}$ and $\frac{3n}{16}$ points. By the union bound, the probability that at least one octant falls outside this range is less than $16e^{-\frac{n^{1/10}}{128}}$. If any octant falls outside this range, we repeat the sampling. The expected number of times we would have to repeat is less than $2$.
\end{proof}

Recursive application of this method gives us the following lemma which corresponds to Lemma 3.3 in \cite{afshani} in the case where the points are in $\mathbb{R}^3$. 
\begin{lemma}\label{lem:partition}
    For any set of points $Q$ in $\mathbb{R}^3$, we can partition $Q$ into $r$ subsets $Q_1,\dots,Q_r$, each of size $\Theta(n/r)$ and find $r$ polyhedral cells $\gamma_1,\dots,\gamma_r$ each with $O(\log r)$ faces such that $Q_i$ is contained in $\gamma_i$ and every plane intersects with $O(r^{\log_87})$ cells. We can do this in $O(n\log r)$ time in expectation.
\end{lemma}
\newpage
\section{3D Convex Hull Algorithm}\label{3dalg}

For completeness we present here \texttt{hull3d}, the algorithm of Afshani et al.~\cite{afshani} for three-dimensional convex hulls, but replacing their partitioning subroutine by our linear-time approximate eight-partition from \cref{sec:3dhull}.

\begin{algorithm}
    \caption{Given a set $Q$ of $n$ points in $\mathbb{R}^3$, the algorithm finds the convex hull of the points.}
    \textbf{Algorithm} \texttt{hull3d}($Q$):
    \begin{algorithmic}[1]
        \FOR{$j=0,1,\dots,\lfloor\log(\delta\log n)\rfloor$}
        \STATE Partition $Q$ into $r_j=2^{2^j}$ subsets $Q_1,\dots,Q_{r_j}$ and cells $\gamma_1,\dots,\gamma_{r_j}$ by \cref{lem:partition}.
        \FOR{$i=1,\dots,r_j$}
        \STATE if $\gamma_i$ is strictly below the upper hull of $Q$ then prune all points in $Q_i$ from $Q$.
            \ENDFOR
        \ENDFOR
        \STATE Compute the upper hull of the remaining points directly.
    \end{algorithmic}
\end{algorithm}
Steps 3 and 4 are done using Lemma 3.8 from \cite{afshani}, which refers to \cite{chan1996optimal,chan1996output,chan1996fixed}. Step 5 is done using any worst case optimal convex hull algorithm which takes $O(n\log n)$ time.

\end{document}

%% file: intro.tex
\section{Introduction}
\emph{Beyond worst-case algorithm design} is directed at designing algorithms whose running
time is asymptotically the best possible 
relative to some metric
of the input instance~\cite{roughgarden2019beyond,roughgarden2021beyond}, 
and there has been considerable work done, for example,
on instance-optimal sorting algorithms;
see, e.g.,~\cite{munro18,juge24,gelling,auger2019,buss19,%
takaoka2009partial,schou2024persisort,BARBAY2013109}.
Focusing on the sorting problem for a moment,
let $X=[x_1,x_2,\ldots,x_n]$ be an input sequence of distinct elements that 
come from a total order, 
and let $\mathcal{R}=\{R_1,R_2,\ldots,R_{\rho(X)}\}$ 
be a division of~$X$ into a set of maximal increasing or decreasing runs
(i.e., consecutive elements in $X$).
These previous papers define a type of
\emph{entropy}, $\Eta({X})$,
for a set of elements, $X$, divided into runs as 
follows~\cite{munro18,juge24,gelling,auger2019,buss19,%
takaoka2009partial,schou2024persisort,BARBAY2013109}:
\[
\Eta({X}) = - \sum_{i=1}^{\rho(X)} 
               \left(\frac{|R_i|}{n}\right) \log \left(\frac{|R_i|}{n}\right).
\]
Sequential comparison-based sorting has a runtime lower bound of $\Omega(n(1+\Eta({X})))$,
and many instance-optimal sorting algorithms 
have running time $O(n(1+\Eta({X})))$,
which can be as small as $O(n)$ depending on the input instance.
For example, 
Auger, Jug{\'e}, Nicaud, and Pivoteau~\cite{auger2019} show that 
the popular TimSort algorithm has this time, and similarly
efficient stack-based mergesort algorithms have been given by
Munro and Wild~\cite{munro18},
Jug{\'e}~\cite{juge24},
Gelling, Nebel, Smith, and Wild~\cite{gelling},
Buss and Knop’~\cite{buss19},
Takaoka~\cite{takaoka2009partial},
and Barbay and Navarro~\cite{BARBAY2013109}.
At CCCG'24,
Schou and Wang introduce PersiSort~\cite{schou2024persisort},
which also has this running time.
Still, we are not aware of any work on geometric problems beyond sorting that concern instance optimality 
with respect to run-based entropy.

%
Instead, given an input, $S$, to a geometric problem,
such as convex hull or maxima set,
Afshani, Barbay, and Chan~\cite{afshani} introduce the
\emph{structural entropy}, $\Eta(S)$ for $S$,
and they provide algorithms that run in
$O(n(1+\Eta(S)))$ time for several such problems, where $n$
is the size of the input, including 2D maxima and 2D and 3D convex hulls.
In contrast to the entropy used for the sorting problem,
however, structural entropy explicitly ignores any near-sortedness
in the input, such as can occur, e.g., with 
convex polygons or monotone polygonal chains.

\paragraph{Our Results.}
In this paper,
we introduce a unification of the entropy used for sorting and the structural 
entropy of Afshani {\it et al.}~\cite{afshani}, which we call 
\emph{range-partition entropy}. 
This unification applied even for problems that do not have obvious
structural entropy: for instance, in the visibility and lower-envelope problems
we study, 
range-partition entropy
equals the entropy of an embedded sorting problem.
In other cases, such as for 3D convex hulls, sortedness provides no obvious
advantage, and case
range-partition entropy
reduces to the structural entropy of Afshani
{\it et al}.

We provide simple algorithms whose running times
depend on the range-partiton entropy,
$\Eta(S)$, of an input set of points, $S$.
For example, we give algorithms for computing maxima sets
and convex hulls for $n$ points in the plane in
$O(n(1+\Eta(S)))$ time. These algorithms perform no worse than any other algorithm for any input on their respective worst permutations, while also taking advantage of sortedness in the input.
We also 
show how to adapt any instance-optimal natural mergesort to compute the lower
envelope of monotone polygonal chains or the visibility polygon
of a point inside a convex polygon with convex holes in
time that depends on the 
range-partition entropy.
In addition, we give a simple randomized algorithm for computing
the convex hull of $n$ points in $\R^3$ and show its expected running time
to be $O(n(1+\Eta(S)))$.

%% file: main.bbl
\begin{thebibliography}{10}

\bibitem{afshani}
Peyman Afshani, J\'{e}r\'{e}my Barbay, and Timothy~M. Chan.
\newblock Instance-optimal geometric algorithms.
\newblock {\em Journal of the ACM}, 64(1):A3:1--A3:38, March 2017.
\newblock \href {https://doi.org/10.1145/3046673} {\path{doi:10.1145/3046673}}.

\bibitem{auger2019}
Nicolas Auger, Vincent Jug{\'e}, Cyril Nicaud, and Carine Pivoteau.
\newblock On the worst-case complexity of {TimSort}, 2019.
\newblock Previously announced at ESA 2018.
\newblock \href {https://arxiv.org/abs/1805.08612v3} {\path{arXiv:1805.08612v3}}.

\bibitem{BARBAY2013109}
J{\'e}r{\'e}my Barbay and Gonzalo Navarro.
\newblock On compressing permutations and adaptive sorting.
\newblock {\em Theoretical Computer Science}, 513:109--123, 2013.
\newblock \href {https://doi.org/10.1016/j.tcs.2013.10.019} {\path{doi:10.1016/j.tcs.2013.10.019}}.

\bibitem{buss19}
Sam Buss and Alexander Knop.
\newblock Strategies for stable merge sorting.
\newblock In {\em ACM-SIAM Symposium on Discrete Algorithms (SODA)}, pages 1272--1290. SIAM, 2019.
\newblock \href {https://doi.org/10.1137/1.9781611975482.78} {\path{doi:10.1137/1.9781611975482.78}}.

\bibitem{chan1996fixed}
Timothy~M Chan.
\newblock Fixed-dimensional linear programming queries made easy.
\newblock In {\em Proceedings of the Twelfth Annual Symposium on Computational Geometry}, pages 284--290. Association for Computing Machinery, 1996.
\newblock \href {https://doi.org/10.1145/237218.237397} {\path{doi:10.1145/237218.237397}}.

\bibitem{chan1996optimal}
Timothy~M. Chan.
\newblock Optimal output-sensitive convex hull algorithms in two and three dimensions.
\newblock {\em Discrete \& Computational Geometry}, 16(4):361--368, 1996.
\newblock \href {https://doi.org/10.1007/BF02712873} {\path{doi:10.1007/BF02712873}}.

\bibitem{chan1996output}
Timothy~M. Chan.
\newblock Output-sensitive results on convex hulls, extreme points, and related problems.
\newblock {\em Discrete \& Computational Geometry}, 16(4):369--387, 1996.
\newblock \href {https://doi.org/10.1007/BF02712874} {\path{doi:10.1007/BF02712874}}.

\bibitem{gelling}
William~Cawley Gelling, Markus~E. Nebel, Benjamin Smith, and Sebastian Wild.
\newblock Multiway powersort.
\newblock In {\em Symposium on Algorithm Engineering and Experiments (ALENEX)}, pages 190--200. SIAM, 2023.
\newblock \href {https://doi.org/10.1137/1.9781611977561.ch16} {\path{doi:10.1137/1.9781611977561.ch16}}.

\bibitem{GRAHAM1972132}
R.~L. Graham.
\newblock An efficient algorith for determining the convex hull of a finite planar set.
\newblock {\em Information Processing Letters}, 1(4):132--133, 1972.
\newblock \href {https://doi.org/10.1016/0020-0190(72)90045-2} {\path{doi:10.1016/0020-0190(72)90045-2}}.

\bibitem{Hoeffding01031963}
Wassily Hoeffding.
\newblock Probability inequalities for sums of bounded random variables.
\newblock {\em Journal of the American Statistical Association}, 58(301):13--30, 1963.
\newblock \href {https://doi.org/10.1080/01621459.1963.10500830} {\path{doi:10.1080/01621459.1963.10500830}}.

\bibitem{juge24}
Vincent Jug\'{e}.
\newblock Adaptive shivers sort: An alternative sorting algorithm.
\newblock {\em ACM Trans. Algorithms}, 20(4):31:1--31:55, August 2024.
\newblock Previously announced at SODA 2020.
\newblock \href {https://doi.org/10.1145/3664195} {\path{doi:10.1145/3664195}}.

\bibitem{10.1145/323233.323246}
David~G. Kirkpatrick and Raimund Seidel.
\newblock Output-size sensitive algorithms for finding maximal vectors.
\newblock In {\em 1st ACM Symposium on Computational Geometry (SoCG)}, page 89–96, 1985.
\newblock \href {https://doi.org/10.1145/323233.323246} {\path{doi:10.1145/323233.323246}}.

\bibitem{kirkpatrickseidelconvex}
David~G. Kirkpatrick and Raimund Seidel.
\newblock The ultimate planar convex hull algorithm?
\newblock {\em SIAM Journal on Computing}, 15(1):287--299, 1986.
\newblock \href {https://arxiv.org/abs/https://doi.org/10.1137/0215021} {\path{arXiv:https://doi.org/10.1137/0215021}}, \href {https://doi.org/10.1137/0215021} {\path{doi:10.1137/0215021}}.

\bibitem{Avis1992}
Ji{\v{r}}{\'\i} Matou{\v{s}}ek.
\newblock Efficient partition trees.
\newblock {\em Discrete \& Computational Geometry}, 8:315--334, 1992.
\newblock \href {https://doi.org/10.1007/bf02293051} {\path{doi:10.1007/bf02293051}}.

\bibitem{munro18}
J.~Ian Munro and Sebastian Wild.
\newblock Nearly-optimal mergesorts: fast, practical sorting methods that optimally adapt to existing runs.
\newblock In Yossi Azar, Hannah Bast, and Grzegorz Herman, editors, {\em 26th European Symposium on Algorithms (ESA)}, volume 112 of {\em LIPIcs}, pages 63:1--63:16. Schloss Dagstuhl, 2018.
\newblock \href {https://doi.org/10.4230/LIPIcs.ESA.2018.63} {\path{doi:10.4230/LIPIcs.ESA.2018.63}}.

\bibitem{peters2015timsort}
Tim Peters.
\newblock listsort.txt, March 21, 2024.
\newblock URL: \url{https://github.com/python/cpython/blob/main/Objects/listsort.txt}.

\bibitem{roughgarden2019beyond}
Tim Roughgarden.
\newblock Beyond worst-case analysis.
\newblock {\em Communications of the ACM}, 62(3):88--96, 2019.
\newblock \href {https://doi.org/10.1145/3232535} {\path{doi:10.1145/3232535}}.

\bibitem{roughgarden2021beyond}
Tim Roughgarden.
\newblock {\em Beyond the Worst-case Analysis of Algorithms}.
\newblock Cambridge University Press, 2021.
\newblock \href {https://doi.org/10.1017/9781108637435} {\path{doi:10.1017/9781108637435}}.

\bibitem{schou2024persisort}
Jens Kristian~Refsgaard Schou and Bei Wang.
\newblock {PersiSort}: a new perspective on adaptive sorting based on persistence.
\newblock In {\em Canadian Conference on Computational Geometry (CCCG)}, pages 287--312, 2024.

\bibitem{takaoka2009partial}
Tadao Takaoka.
\newblock Partial solution and entropy.
\newblock In Rastislav Kr{\'{a}}lovic and Damian Niwinski, editors, {\em Mathematical Foundations of Computer Science 2009, 34th International Symposium, {MFCS} 2009, Novy Smokovec, High Tatras, Slovakia, August 24-28, 2009. Proceedings}, volume 5734 of {\em Lecture Notes in Computer Science}, pages 700--711. Springer, 2009.
\newblock \href {https://doi.org/10.1007/978-3-642-03816-7_59} {\path{doi:10.1007/978-3-642-03816-7_59}}.

\bibitem{octants}
F.~Frances Yao, David~P. Dobkin, Herbert Edelsbrunner, and Michael~S. Paterson.
\newblock Partitioning space for range queries.
\newblock {\em SIAM Journal on Computing}, 18(2):371--384, 1989.
\newblock \href {https://doi.org/10.1137/0218025} {\path{doi:10.1137/0218025}}.

\end{thebibliography}
